\documentclass[11pt,fleqn]{amsart}

\newtheorem{lemma}{Lemma}[section]

\newtheorem{corollary}[lemma]{Corollary}

\newtheorem{assumption}[lemma]{Assumption}

\usepackage{amsmath,amssymb,amsfonts}
\usepackage[round]{natbib}
\usepackage{dsfont}
\usepackage{enumerate}
\usepackage{mathrsfs}
\usepackage{graphicx}
\usepackage{nicefrac}
\usepackage[ansinew]{inputenc}
\usepackage{EGM-abbreviations}
\usepackage{pifont} 

\begin{document}
\title[Elliptical graphical modelling]{Elliptical graphical modelling}

\author[D. Vogel]{Daniel Vogel}
\author[R. Fried]{Roland Fried}

\address{Fakult\"at Statistik, Technische Universit\"at Dortmund, 
44221 Dortmund, Germany}
\email{daniel.vogel@tu-dortmund.de} \email{fried@statistik.tu-dortmund.de}


\begin{abstract}
We propose elliptical graphical models based on conditional uncorrelatedness as a generalization of Gaussian graphical models by letting the population distribution be elliptical instead of normal, allowing the fitting of data with arbitrarily heavy tails. 
We study the class of proportionally affine equivariant scatter estimators and show how they can be used to perform elliptical graphical modelling, leading to a new class of partial correlation estimators and analogues of the classical deviance test. 
General expressions for the asymptotic variance of partial correlation estimators, unconstrained and under decomposable models, are given, and the asymptotic chi square approximation of the pseudo-deviance test statistic is proved.
The feasibility of our approach is demonstrated by a simulation study, using, among others, Tyler's scatter estimator, which is distribution-free within the elliptical model. 
Our approach provides a robustification of Gaussian graphical modelling. The latter is likelihood-based and known to be very sensitive to model misspecification and outlying observations. 
\end{abstract}

\maketitle

\section{Introduction and notation}  

The statistical theory of undirected graphical models for continuous variables is usually based on the assumption of multivariate normality. 
In practice, data may deviate from the normal model in various ways. Outliers and heavy tails pose a problem of particular gravity: they frequently occur, and the normal likelihood methods, such as the sample covariance matrix, are very susceptible to them. Our objective is to deal with heavy-tailed data and to safeguard graphical modelling against the impact of faulty outliers. We restrict our attention to the case where we have only continuous variables and only undirected edges. Joint multivariate normality is often assumed in this situation, and the statistical methodology is called Gaussian graphical modelling.  
We propose the class of elliptical distributions as a more general model and call our approach elliptical graphical modelling.

The lack of robustness of Gaussian graphical modelling has been noted by several authors. Four proposals of robust approaches to Gaussian graphical modelling are known to us:
\citet{Becker2005} and \citet{Gottard2010} suggest replacing the sample covariance matrix by the reweighted minimum covariance determinant estimator. \citet{Miyamura2006} propose an alternative M-type estimation, and Finegold \& Drton (arXiv:1009.3669) consider robustified versions of the graphical lasso by \citet{Friedman2008}.

This article delivers a systematic treatment of the plug-in approach used in the first two references. We show that the sample covariance matrix may be replaced by any affine equivariant, root-$n$-consistent estimator. As long as ellipticity can be assumed, the classical Gaussian graphical modelling tools can be employed with simple adjustments. Thus the data analyst is free to choose the appropriate estimator, delivering the degree of robustness necessary for the data situation at hand. In order to reduce the search space, graphical modelling is often restricted to decomposable graphical models, which allow better interpretability, cf.\ \citet[][Chapter~12]{Whittaker1990}, but are also easier to handle mathematically. For conciseness we restrict our derivations to decomposable models.

We close this section by introducing some mathematical notation. Depending on the context, the symbol $\sim$ means distributed as or asymptotically equivalent. Finite index sets are denoted by small Greek letters. Subvectors and submatrices are referenced by subscripts, e.g.\ for $\alpha, \beta \subseteq \{1,...,p\}$ the $|\alpha|\times|\beta|$ matrix $S\!_{\alpha,\beta}$ is obtained from $S$ by deleting all rows that are not in $\alpha$ and all columns that are not in $\beta$. Similarly, the $p \times p$ matrix $(S\!_{\alpha,\beta})^{(p)}$ is obtained from $S$ by putting all rows not in $\alpha$ and all columns not in $\beta$ to zero.
We view this matrix operation as two operations performed sequentially: first $(\cdot)_{\alpha,\beta}$ extracting the submatrix and then $(\cdot)^{(p)}$ writing it back  on a blank matrix at the coordinates specified by $\alpha$ and $\beta$. Of course, the latter is not well defined without the former, but this allows us to write $(S_{\!\alpha,\beta}^{-1})^{(p)}$, for example. Subscripts have priority over superscripts, $S_{\!\alpha,\beta}^{-1}$ stands for $(S\!_{\alpha,\beta})^{-1}$.
Let $\Ss_p$ and $\Ss^+_p$ be the sets of all symmetric, respectively positive definite $p \times p$ matrices, and define  $A_D$ as the diagonal matrix having the same diagonal as $A \in \Rpp$. 
The Kronecker product $A \otimes B$ of two matrices $A,B \in \R^{p \times p}$ is defined as the $p^2 \times p^2$ matrix with entry
$a_{i,j} b_{k,l}$ at position $\{(i-1)p + k, (j-1)p + l\}$. Let $e_1, \ldots, e_p$ be the unit vectors in $\R^p$ and $1_p$ the $p$-vector consisting only of ones. Define  the matrices:
\[
	J_p = \sum_{i=1}^p e_i^{} e_i^T \otimes e_i^{} e_i^T, \qquad \
  K_p = \sum_{i=1}^p \sum_{j=1}^p e_i^{} e_j^T \otimes e_j^{} e_i^T, \qquad
  M_p = \frac{1}{2}\left( I_{p^2} + K_p \right),
\]
where $I_{p^2}$ denotes the $p^2 \times p^2$ identity matrix; $K_p$ is also called the commutation matrix. Finally, let $\vec \,(A)$ be the $p^2$-vector obtained by stacking the columns of $A \in \R^{p \times p}$ from left to right underneath each other. More on these concepts and their properties can be found in \citet{Magnus1999}.

\section{Elliptical graphical models}

We introduce elliptical graphical models in analogy to Gaussian graphical models. For details on the latter see \citet{Whittaker1990}, \citet{Cox1996}, \citet{Lauritzen1996} or \citet{Edwards2000}.

Consider the class $\Ee_p$ of all continuous, elliptical distributions on $\R^p$. A continuous distribution $F$ on $\mathds{R}^p$ is said to be elliptical if it has a density $f$ of the form
\begin{equation} \label{density}
	f(x) = \det(S)^{-1/2} g\big\{(x-\mu)^T S^{-1} (x-\mu)\big\}
\end{equation}
for some $\mu \in \mathds{R}^p$ and symmetric, positive definite $p \times p$ matrix $S\!$. We call  $S$ the shape matrix of $F$, and denote the class of all continuous elliptical distributions
on $\mathds{R}^p$ with the parameters $\mu$ and $S$ by $\mathscr{E}_p(\mu, S)$. A continuous distribution on $\R^p$ is called spherical if $S$ is proportional to the identity matrix. The shape matrix $S$ is unique only up to scale,
that is, $\Ee_p(\mu, S) = \Ee_p(\mu, c S)$ for any $c > 0$. Several forms of standardization have been suggested in the literature. \cite{Paindaveine2008} argues for $\det(S) = 1$. For our considerations the standardization of $S$ is irrelevant, and we understand the shape of an elliptical distribution as an equivalence class of positive definite random matrices being proportional to each other and call any matrix $S\!$ satisfying (\ref{density}) for a suitable function $g$ a shape matrix of $F$. We likewise view its inverse $K = S^{-1}$, which we call a pseudo concentration matrix of $F$. Furthermore let
\[
	h: \Ss^+_p \to \Ss_p: A \mapsto - \left(A^{-1}\right)_D^{-1/2} A^{-1} \left(A^{-1}\right)_D^{-1/2}
\] 
and $P = h(S)$.
The function $h$ is invariant to scale changes, i.e., $P$ is a uniquely defined parameter of $F \in \Ee_p(\mu,S)$. The diagonal elements of $P$ are equal to $-1$.
If the second-order moments of $X \sim F \in \Ee_p(\mu,S)$ exist, then $\Sigma = \var(X)$ is proportional to $S$. Consequently, the element $p_{i,j}$ of $P$ at position $(i,j)$ is the partial correlation of $X_i$ and $X_j$ given the other components of $X$ \citep[][Chapter~5]{Whittaker1990}. We call $P$ the generalized partial correlation matrix of $F$ and refer to it as partial correlation matrix for brevity.

The qualitative information of $P$ can be coded in an undirected graph $G = (V,E)$, where $V$ is the vertex set and $E$ the edge set, in the following way: the variables $X_1,\ldots,X_p$ are the vertices, and an edge is drawn between $X_i$ and $X_j$ if and only if $p_{i,j} \neq 0$ $(i,j = 1,\ldots,p; i \neq j)$. The graph $G$ thus obtained is called the generalized partial correlation graph of $F$. Formally we set $V = \{1,\ldots,p\}$ and write the elements of $E$ as unordered pairs $\{i,j\}$  $(i,j = 1,\ldots,p; i \neq j)$. The global and the local Markov property with respect to any generalized partial correlation graph $G$ are equivalent for any $F \in \Ee_p$ without any moment assumptions \citep{Vogel2010a}.

Let $\Ss^+_p(G)$ be the subset of $\Ss^+_p$ consisting of all positive definite matrices with zero entries at the positions specified by the graph $G = (V,E)$, i.e.,
\[
   K  \in \Ss^+_p(G) \ \ \Longleftrightarrow \ \   K  \in \Ss^+_p, \ \ k_{i,j} = 0 \  \ (i \neq j,\ \{i,j\} \notin E),
\]
and define
\[ 
   \Ee_p(G) = \left\{ \ F \in \Ee_p(\mu,K^{-1}) \ \middle| \ \mu \in \R^p, \  K \in \Ss^+_p(G)\  \right\}
\]
to be the elliptical graphical model induced by $G$. We call the model $\Ee_p(G)$ decomposable if $G$ is decomposable, i.e., if it possesses no chordless cycle of length greater than three. For alternative characterizations and properties of decomposable graphs see e.g.\ \citet[][Chapter~2]{Lauritzen1996}. 

In the remainder of this section we discuss the interpretation of an absent edge in the partial correlation graph
of $F \in \Ee_p$. Let us assume that the second-order moments of $X \sim F$ are finite.
The partial uncorrelatedness of, say, $X_1$ and $X_2$ given $X_3,\ldots,X_p$, i.e., $p_{1,2} = 0$, is to be understood as linear independence of $X_1$ and $X_2$ after the common linear effects of $X_3,\ldots,X_p$ have been removed.
A relation of similar type is conditional independence: roughly, $X_1$ and $X_2$ are conditionally independent given $X_3,\ldots,X_p$, if the conditional distribution of $(X_1,X_2)$ is a product measure for almost all values of the conditioning variable $(X_3,\ldots,X_p)$. In comparison to partial correlation we understand conditional independence as complete independence of $X_1$ and $X_2$ after the removal of all common effects of $X_3,\ldots,X_p$.

Another related term is conditional uncorrelatedness: the conditional distribution of $(X_1,X_2)$ given $(X_3,\ldots,X_p)$ has correlation zero for almost all values of $(X_3,\ldots,X_p)$. There is an important qualitative difference between partial and conditional correlation: the former is a real value, the latter a function of the conditioning variable.
All marginal and conditional distributions of elliptical distributions are again elliptical \citep[][Section~2.6]{Fang1990b}. Hence partial uncorrelatedness implies conditional uncorrelatedness \citep{Baba2004}, and $p_{1,2} = 0$ means linear independence of $X_1$ and $X_2$ after all common effects of $X_3,\ldots,X_p$ have been removed.

However, the only spherical distributions with independent margins are Gaussian distributions, cf.\ \citet[][p.~51]{Bilodeau1999}. Thus contrary to Gaussian graphical models a missing edge in the partial correlation graph of an elliptical distribution can in general not be interpreted as conditional independence. It appears, that by going from the normal to the elliptical model, the gain in generality is paid by a loss in the strength of inference. But this loss is illusory.
From a data modelling perspective the conditional independence interpretation of partial uncorrelatedness under normality is an assumption, not a conclusion. By modelling multivariate data by a joint Gaussian distribution one models the linear dependencies and assumes that there are no other than linear associations among the variables. By fitting an appropriate non-Gaussian model one may still model the linear dependencies and allow non-linear dependencies. Using semiparametric models embodies this idea: the aspects of interest, in our case linear dependencies, are modelled parametrically, whereas other aspects remain unspecified.

Of course, non-normal data need not be elliptical. Any relevant data feature, such as non-linearities, anomalous values, etc., is of potential interest and should be analysed. If the data, say, contains strong quadratic interactions, models that incorporate them should be used, as it is described e.g.\ in \citet[][Section 2.10]{Cox1996}. We address primarily the situation where the essential structure of the data is captured by an ellipse, and the linear interactions are the prominent ones. In any case, a robust analysis of the linear effects, as proposed here, is a suitable starting point of any subsequent tests for potential non-linear effects.

\section{Unconstrained estimation}

An important initial step towards elliptical graphical modelling is the unconstrained estimation of $P$. Unconstrained, since we do not assume a graphical model to hold, not forcing any constraints 
on $P$. We will consider estimators of the type $\hP\!_n = h(\hS\!_n)$, where $\hS\!_n$ is a suitable estimator of a multiple of $S$, therefore start by considering shape estimators $\hS\!_n$.

Let $X_1,\ldots,X_n$ be independent and identically distributed random vectors sampled from an elliptical distribution $F \in \Ee_p(\mu, S)$. Depending on the context, $X_k$ may denote the $k$th $p$-dimensional observation or the $k$th component of the vector $X$. Furthermore let $\X_n = (X_1,\ldots,X_n)^T$ be the $n \times p$ data matrix and $\hS\!_n = \hS\!_n(\X_n)$ be a scatter estimator. The symbol $\hS\!_n$ may have two meanings: a function on the sample space, or as abbreviation for $\hS\!_n(\X_n)$, a random variable. We use the term scatter estimator for any symmetric matrix-valued estimator that gives some information about the spread of the data. 
We call $\hS\!_n$ affine pseudo-equivariant, if it satisfies 
\be \label{ape}
	\hS\!_n(\X_n A^T + 1_n b^T ) \propto A \hS\!_n(\X_n) A^T
\ee
for all $b \in \R^p$ and full rank $A \in \R^{p \times p}$.
This is a generalization  of the strict affine equivariance for scatter estimators, which is obtained if (\ref{ape}) is satisfied with equality. We use this weaker condition since overall scale is irrelevant for partial correlations, and we want to include estimators which only estimate shape, but not scale, and do not satisfy strict affine equivariance. Examples are given in Section~\ref{sec:Examples}.

\citet{Tyler1982} shows that, if a strictly affine equivariant scatter estimator is evaluated at an elliptical distribution, its first two moments, if existent, have a common structure. If the proportionality factor in (\ref{ape}) is not random, the same holds true for pseudo-equivariant scatter estimators. The following condition is therefore natural for affine pseudo-equivariant estimators at elliptical distributions $F$, and many shape estimators have been shown to satisfy it under suitable additional conditions on $F$, see also the examples in Section~\ref{sec:Examples}. 
\begin{assumption} \label{asymS} 
The estimator $\hS\!_n$ converges in probability to $\eta S$ for some $\eta \ge 0$, and there exist $\sigma_1 \ge 0$ and $\sigma_2 \ge - 2 \sigma_1 / p$ such that
\[	
	n^{1/2} \vec\,( \hS\!_n - \eta S ) \to N_{p^2}\left\{ 0, \eta^2 W_S(\sigma_1,\sigma_2)\right\} \hfill
\]
in distribution as $n \to \infty$, where $W_S(\sigma_1,\sigma_2) = 2 \sigma_1 M_p(S \otimes S) + \sigma_2 \vec S (\vec S)^T$.
The scalars $\sigma_1$ and $\sigma_2$ depend on the estimator $\hS\!_n$, the dimension $p$ and the function $g$, but are constant with respect to the shape $S$.
\end{assumption}
We have the following implication for the derived estimators $\hK_n = \hS\!_n^{-1}$ and $\hPhS\!_n = h(\hS\!_n)$. 
\bp \label{Prop1}
If $\hS\!_n$ satisfies Assumption \ref{asymS}, then with $K = S^{-1}$,
\begin{enumerate}[(i)]
\item \label{asymK}
$	\displaystyle n^{1/2} \vec\,( \hK_n - \eta^{-1} K ) \to N_{p^2}\left\{0,\eta^{-2} W_K(\sigma_1,\sigma_2)\right\}$\\[.5ex]
in distribution as $n \to \infty$, where $W_K(\sigma_1,\sigma_2) = 2 \sigma_1 M_p(K \otimes K) + \sigma_2 \vec K (\vec K)^T$, and
\vspace{.5ex}
\item \label{asymP}
$ \displaystyle n^{1/2} \vec\,( \hPhS\!_n - P ) \to N_{p^2}\left\{0, \ 2 \sigma_1 \Gamma(S) M_p (K \otimes K) \Gamma(S)^T\right\}$\\[.3ex]
in distribution  as $n \to \infty$ with $\Gamma(S) = (K_D^{-1/2} \otimes K_D^{-1/2}) + M_p (P \otimes K_D^{-1}) J_p$.
\end{enumerate}
\ep
An important aspect of Proposition \ref{Prop1} is that under ellipticity the asymptotic covariance matrices of partial correlation estimators $\hPhS\!_n$ derived from affine equivariant shape estimators $\hS\!_n$ are proportional to each other.

\section{Constrained estimation}

In this section we treat the estimation of $P$ under a given graphical model $\Ee_p(G)$ specified by the graph $G = (V,E)$, i.e., estimating $P$ with zero-entries. A crude approach is to put the concerning elements in an unconstrained estimate $\hP\!_n$ to zero, but this generally destroys the positive definiteness of the estimate. We define the function 
$h_G: \Ss^+_p \to \Ss^+_p(G): A \mapsto A_G$ by
\be \label{h_G}
\begin{cases}
 \  (A_G)_{i,j} = a_{i,j} 		&	\qquad    \left(\{i,j\} \in E \  \vee  \ i = j\right), \\
 \  (A_G^{-1})_{i,j} = 0     & \qquad    \left(\{i,j\} \notin E,  \ i \neq j\right), \\
\end{cases}
\ee  
where $a_{i,j}$ are the elements of $A$.
A unique and positive definite solution $A_G$ of (\ref{h_G}) exists for any positive definite $A$. The positive definiteness of $A$ is sufficient but not necessary. For details see \citet[][p.~133]{Lauritzen1996}. Since we mainly deal with asymptotics, and shape estimators $\hS\!_n$ are usually almost surely positive definite at continuous distributions for sufficiently large $n$, we assume positive definiteness for simplicity's sake.

Let $G = (V,E)$ be a decomposable graph with cliques $\gamma_1,\ldots,\gamma_c$ ($c \ge 1$), and 
define the sequence $\delta_1,\ldots,\delta_{c-1}$ of successive intersections by
\[
	\delta_k = (\gamma_1 \cup \cdots \cup \gamma_k) \cap \gamma_{k+1} \qquad (k = 1,\ldots, c-1).
\]
We assume that the ordering $\gamma_1,\ldots,\gamma_k$ is such that the cliques form a perfect sequence, i.e., for all $k = 1,\ldots,c-1$ there is a $j \in \{1,\ldots,k\}$ such that $\delta_k \subseteq \gamma_j$. It is always possible to arrange the cliques of a decomposable graph in a perfect sequence \citep[][Prop.~2.17]{Lauritzen1996}.
For notational convenience we let 
\[
	\alpha_k = 
	\begin{cases}
	\gamma_k 			& \qquad (k = 1, \ldots, c), \\
	\delta_{k-c} 	& \qquad (k = c + 1 ,\ldots, 2c - 1),	
	\end{cases}
\qquad \ \qquad 
	\zeta_k = 
	\begin{cases}
	1 			& \qquad (k = 1, \ldots, c), \\
	-1 			& \qquad (k = c + 1 ,\ldots, 2c - 1).	
	\end{cases}
\]
Then $h_G(A)$ allows the following explicit formulation for decomposable G,
\[ 
	h_G(A) = A_G = \left\{ \sum_{k=1}^{2c-1} \zeta_k	\left(A_{\alpha_k,\alpha_k}^{-1} \right)^{(p)}\right\}^{-1} 
	\qquad (A \in \Ss^+_p).
\]
We will use this representation of $h_G$ to further analyse the properties of the estimators $\hS\!_G = h_G(\hS\!_n)$, $\hK\!_G = \hS\!_G^{-1}$ and $\hPSG = h(\hS\!_G)$ for a decomposable graph $G$.
Using the notation $S\!_G = h_G(S)$, $K^{}_G = S\!_G^{-1}$, $P_G = h(S\!_G) \in \R^{p \times p}$ and
\[
	\Omega_G(S) = \sum_{k=1}^{2c-1} \zeta_k	\left(S\!_{\alpha_k,\alpha_k}^{-1} \right)^{(p)} \otimes 	
	\left(S\!_{\alpha_k,\alpha_k}^{-1} \right)^{(p)} \ \in \ \R^{p^2 \times p^2} 
\]
we have the following result about the asymptotic distribution. It is not assumed that the true shape $S$ fits the model $G$.
\bp \label{Prop2}
If $\hS\!_n$ fulfils Assumption \ref{asymS} and $G$ is decomposable, then
\begin{enumerate}[(i)]
\item \label{asymKG}
$\displaystyle n^{1/2} \vec\,( \hK\!_G - \eta^{-1} K_G ) \to N\!_{p^2}\{0, \eta^{-2} W\!_{K_G}(\sigma_1,\sigma_2)\}$ \ \ in distribution \\[.5ex] 
as $n \to \infty$ with \ $W_{K_G}(\sigma_1,\sigma_2) = 2 \sigma_1 M_p \Omega_G(S) (S \otimes S)\Omega_G(S) + \sigma_2 \vec K_G (\vec K_G)^T$,

\smallskip
\item \label{asymSG}
$\displaystyle n^{1/2} \vec\,( \hS\!_G - \eta S\!_G ) \to N\!_{p^2}\{0, \eta^{2} W\!_{S\!_G}(\sigma_1,\sigma_2)\}$ \ \ 
in distribution as $n \to \infty$\\[.5ex]
with $W\!_{S\!_G}(\sigma_1,\sigma_2) = 2 \sigma_1 M_p \left( S\!_G \otimes S\!_G \right) \Omega_G(S) (S \otimes S)\Omega_G(S) \left( S\!_G \otimes S\!_G \right)  +  \sigma_2 \vec S\!_G (\vec S\!_G)^T$, 

\smallskip
\item \label{asymPG}
$\displaystyle n^{1/2} \vec\,( \hP\!_G -  P_G ) \to N\!_{p^2}\{0, W\!_{P_G}(\sigma_1)\}$  \ in distribution as $n \to \infty$, where \\[0.8ex]
$W_{P_G}(\sigma_1) = 2 \sigma_1 \Gamma(S\!_G) M_p \Omega_G(S) (S \otimes S) \Omega_G(S) \Gamma(S\!_G)^T$ with $\Gamma(\cdot)$ as in Proposition \ref{Prop1} (\ref{asymP}).
\end{enumerate}
\ep 
 If the true shape $S$ satisfies the graph $G$, the expressions for the asymptotic variances simplify.
\begin{corollary} \label{Cor1}
If $\hS\!_n$ satisfies Assumption \ref{asymS} with $S^{-1} \in \Ss^+_p(G)$ for a decomposable graph $G$, then the assertions of Proposition \ref{Prop2} are true with 
\begin{enumerate}[(i)]
\item \label{asymKG2}
$W_{K_G}(\sigma_1,\sigma_2) = 2 \sigma_1 M_p \Omega_G(S) \ + \ \sigma_2 \vec K (\vec K)^T$, 
\item \label{asymSG2}
$W\!_{S\!_G}(\sigma_1,\sigma_2) = 2 \sigma_1 M_p (S \otimes S) \Omega_G(S) (S \otimes S) \ + \ \sigma_2 \vec S (\vec S)^T$ and 
\item \label{asymPG2}
$W_{P_G}(\sigma_1) = 2 \sigma_1 \Gamma(S) M_p \Omega_G(S) \Gamma(S)^T$.
\end{enumerate}
\end{corollary}

\section{Testing}
\label{sec:Testing}
An essential tool of most model selection procedures is to test if a model under consideration fits the data and to compare the fit of two nested models.
On the set $\Pi_p = \{ (i,j) \mid i,j = 1,\ldots,p \}$ of the positions of a $p \times p$ matrix we declare a strict ordering $\prec_p$ by 
\[
	(i,j) \prec_p (k,l)\quad  \Leftrightarrow \quad (j-1)p + i < (l-1)p + k, 
	\qquad \quad (i,j,k,l = 1,\ldots,p).  
\]
For any subset $Z = \{ z_1,\ldots,z_q \} \subset \Pi_p$, where $z_k = (i_k, j_k)$ ($k = 1,\ldots,q$) and $z_1 \prec_p \cdots \prec_p z_q$, define the matrix $Q_Z \in \R^{q \times p^2}$ as follows: each line consists of exactly one entry 1 and zeros otherwise. The 1-entry in line $k$ is in column $(i_k-1)p + j_k$. Thus $Q_Z\!\vec \,(A)$ picks the elements of $A$ at positions specified by $Z$ in the order they appear in $\vec \,(A)$. For a graph $G = (V, E)$ with $V = \{ 1,\ldots,p\}$ let
\[
	D(G) = \left\{ (i,j) \, \middle|\, i,j = 1,\ldots,p; \  \{i,j\} \notin E; \ j < i \right\}, 
\]
i.e., the set $D(G)$ gathers all sub-diagonal zero-positions that $G$ enforces on a concentration matrix. Thus $F \in \Ee_p(G)$ is equivalent to $Q_{D(G)} \vec K = 0$.

Now let $G_0 = (V, E_0)$ and $G_1 = (V,E_1)$ be two decomposable graphs with $V$ as above and $E_0 \subsetneq E_1$, or equivalently, $\Ee_p(G_0) \subsetneq \Ee_p(G_1)$. For notational convenience let
\[
	Q_0 = Q_{D(G_0)}, \ \ Q_1 = Q_{D(G_1)}, \ \ Q_{0,1} = Q_{D(G_0)\setminus D(G_1)},
\]	
furthermore
\[
	q_0 = |D(G_0)|, \quad q_1 = |D(G_1)|,   \quad q_{0,1} = q_0 - q_1.
\]
An intuitive approach to testing $G_0$ against the broader model $G_1$ is to reject $G_0$ in favour of $G_1$, if all entries at positions in $D(G_0)\setminus D(G_1)$ of an estimate $\hPhS_{G_1}$ of $P$ under $G_1$ are close to zero. For example, a sum of suitably weighted squared entries of $\hPhS_{G_1}$, such as $\hat{T}\!_n(G_0,G_1)$ below, is a possible test statistic. Let
\[
	R_G(S) = \Gamma(S) M_p \Omega_G(S) \Gamma(S)^T.
\]
For invertible $S$ the matrix $R_{G_1}(S)$ has rank $(p-1)p/2 - q_1$, which can be deduced from the inverse function theorem. Then $Q_{0,1} R_{G_1}(S) Q_{0,1}^T$ is of full rank, and the probability that the Wald-type test statistic 
\[
	\hat{T}\!_n(G_0,G_1) = \frac{n}{2}\left( \vec \hP_{G_1}\right)^T   
															Q_{0,1}^T 
															\left\{ Q_{0,1} R_{G_1}(\hS\!_n) Q_{0,1}^T \right\}^{-1}
															Q_{0,1}
															\vec \hP_{G_1}
\]
exists tends to 1 as $n \to \infty$. Proposition \ref{PropWald} describes the asymptotic behaviour of  $\hat{T}\!_n(G_0,G_1)$ under the null hypothesis that $G_0$ is true, part (\ref{PropWald1}), and under a local alternative, part (\ref{PropWald2}).
\bp \label{PropWald}
Let $G_0$, $G_1$ be as above and $X_1,\ldots,X_n$ independent and identically distributed random variables with $X_1 \sim F \in \Ee_p(\mu,S) \subset \Ee_p(G_0)$. Let $\hS\!_n$ be an affine pseudo-equivariant scatter estimator such that $\hS\!_n(\X_n)$ satisfies Assumption \ref{asymS}.
\begin{enumerate}[(i)]
\item \label{PropWald1}
Then $\hat{T}\!_n(G_0,G_1) \to \sigma_1 \chi^2_{q_{0,1}}$ in distribution as $n \to \infty$.
\item \label{PropWald2}
For $m \in \N$ let $\X_n^{(m)} = (X_1^{(m)},\ldots,X_n^{(m)})^T$ be distributed as $\X_n S\hmh S_m\hh$, thus $X_1^{(m)} \sim \Ee_p(\mu,S_m)$, where the sequence $S_m$ is such that $B = \lim_{m \to \infty} m^{1/2}(S_m - S)$ exists. If, for each $n \in \N$, $\hS\!_n$ is applied to $\X_n^{(n)}$, then, as  $n \to \infty$,  
\be \label{Wald.Alt}
	\hat{T}\!_n(G_0,G_1) \to \sigma_1 \chi^2_{q_{0,1}}\left\{ \sigma_1^{-1}\delta(B,S)\right\}
\ee
in distribution, where
\[
	\delta(B,S) = \frac{1}{2} v^T Q_{0,1}^T \left\{ Q_{0,1} R_{G_1}(S) Q_{0,1}^T\right\}^{-1} Q_{0,1} v,
	\qquad
	v = \Gamma(S) \Omega_{G_1}(S)\vec B.
\] 
\end{enumerate}
\ep
We have some remarks.
\begin{enumerate}[(a)]
\item
We define the non-centrality parameter of the $\chi^2$ distribution $\chi^2_r(\delta) \sim \left(N_r(\mu,I_r)\right)^2$ as $\delta = \mu^T \mu$.
\item
We require $\hS\!_n$ to be affine pseudo-equivariant to ensure that the convergence of $n^{1/2} \{\hS\!_n(\X_n^{(m)}) - \eta S_m\}$ for $n \to \infty$ is uniform in $m$.
\item
In part (\ref{PropWald2}) of Proposition \ref{PropWald} we do not require the sequence of alternatives to lie in the model $G_1$, i.e., that $S\hme_n \in \Ss^+_p(G_1)$, as it is not necessary for the convergence (\ref{Wald.Alt}) to hold. When choosing a model by forward selection one usually compares two wrong models, so it is of interest to know the behaviour of $\hat{T}\!_n(G_0,G_1)$ also if $G_1$ is not true. 
\end{enumerate}
A difficulty with the test in Proposition \ref{PropWald} is the complicated formulation of  $\hat{T}\!_n(G_0,G_1)$. The classical test in Gaussian graphical models is the deviance test. The next proposition gives the analogue for elliptical graphical modelling. It treats parts (\ref{PropWald1}) and (\ref{PropWald2}) of the previous proposition simultaneously. 
\bp \label{PropDeviance}
Let $G_0$, $G_1$ be as above and $\hS\!_n$ a sequence of almost surely positive definite random $p \times p$ matrices, for which $n^{1/2} (\hS\!_n - S)$ converges in distribution to a non-degenerate limit for some $S \in \Ss^+_p$ with $S^{-1} \in \Ss^+_p(G_0)$. Then, as $n \to \infty$, 
\[
	\hat{D}_n(G_0,G_1) = n \left\{ \log \det h_{G_0}(\hS\!_n) - \log \det h_{G_1}(\hS\!_n) \right\}\ \sim \ \hat{T}\!_n(G_0,G_1).
\]
\ep
If the larger model $G_1$ is the saturated model, then Proposition \ref{PropDeviance} is a corollary of Theorem 2 in \citet{Tyler1983}. We extend Tyler's result to two nested models.
\begin{corollary} \label{CorDev}
Both assertions (\ref{PropWald1}) and (\ref{PropWald2}) of Proposition \ref{PropWald} remain true, if $\hat{T}\!_n(G_0,G_1)$ is replaced by  $\hat{D}_n(G_0,G_1)$.
\end{corollary} 

\section{Examples}
\label{sec:Examples}
There are many affine equivariant, robust estimators, see, for example, \citet{Zuo2006} or \cite{Maronna2006}. The comparison of asymptotic properties of such estimators in the elliptical model reduces to a comparison of the respective values of the scalars $\sigma_1$ and $\sigma_2$.
Of course, the sample covariance matrix is affine equivariant. The following can be found in \citet{Tyler1982}.
\bp \label{asymSigma}
If $X_1,\ldots,X_n$ are independent and identically distributed with distribution $F \in \Ee_p(\mu,S)$ and 
$E||X_1 - \mu||^4 < \infty$, then $\hat{\Sigma}_n = \hat{\Sigma}_n(\X_n)$ fulfils Assumption \ref{asymS} with $\sigma_1 = 1 + \kappa/3$ and $\sigma_2 = \kappa/3$, where $\kappa$ is the excess kurtosis of any component of $X_1$.
\ep
Proposition \ref{asymSigma} indicates the inappropriateness of the sample covariance matrix for heavy-tailed distributions: its asymptotic distribution depends on the kurtosis, which is large at heavy-tailed distributions, rendering the estimator inefficient. An alternative is Tyler's M-estimator, which is defined as the solution $\hat{V}\!_n = \hat{V}\!_n(\X_n)$ of 
\[ 
	\frac{p}{n}\sum_{i=1}^{n} 
	\frac{(X_i - \overline{X}_n)(X_i - \overline{X}_n)^T}
		{(X_i - \overline{X}_n)^T \hat{V}\!_n^{-1}(X_i - \overline{X}_n)}
		= \hat{V}\!_n
\]
that satisfies $\det \hat{V}\!_n = 1$. Existence, uniqueness and asymptotic properties are treated in \citet{Tyler1987}, where the following result is proven.
\bp \label{asymTyler}
If $X_1,\ldots,X_n$ are independent and identically distributed with distribution $F \in \Ee_p(\mu,S)$, furthermore 
$E||X_1 - \mu||^2 < \infty$ and $E||X_1 - \mu||^{-3/2} < \infty$, then $\hat{V}\!_n$ fulfils Assumption \ref{asymS} with $\sigma_1 = 1 + 2/p$ and $\sigma_2 = - 2(1 + 2/p)/p$.
\ep
We have the following remarks.
\begin{enumerate}[(a)]
\item
In Proposition \ref{asymSigma} the scalars $\sigma_1$ and $\sigma_2$ are constant, irrespective of the function $g$, i.e., the Tyler matrix is asymptotically distribution-free within the elliptical model. Hence, when carrying out any of the tests from Section \ref{sec:Testing}, $\sigma_1$ does not need to be estimated.
\item
Tyler's matrix can cope with arbitrarily heavy tails. The assumption of finite second moments is only required for location estimation by the mean. It may be replaced by any root-$n$-consistent location estimator, for instance the Hettmansperger--Randles (2002) median. The inverse moment condition  $E||X_1 - \mu||^{-3/2} < \infty$ is fairly mild: for $p \ge 2$ it is fulfilled if  $g$ has no singularity at $0$.
\item 
The estimator $\hat{V}\!_n$ is affine pseudo-equivariant and gives information only about the shape but none about the scale. Other such estimators are Oja sign and rank covariance matrices \citep{Ollila2003, Ollila2004}.
\end{enumerate}
The popular reweighted minimum covariance determinant estimator \citep{Rousseeuw1987,Croux1999} is highly robust and affine equivariant and has previously been proposed in the context of graphical modelling \citep{Becker2005, Gottard2010}. It is defined as follows. 
A subset $\tau \subset \{1,\ldots,n\}$ of size $h = \lceil t n \rceil$, where $1/2 \le t < 1$ is fixed, is determined such that $\det(\hSigma^{\tau})$ is minimal with
\[
	\hSigma^{\tau} = \frac{1}{h} \sum_{i \in \tau} (X_i - \bar{X}^{\tau}) (X_i - \bar{X}^{\tau})^T,
	\qquad \quad
	\bar{X}^{\tau} = \frac{1}{h} \sum_{i \in \tau} X_i.
\]
The mean $\hmu_{\rm MCD}$ and covariance matrix $\hSigma_{\rm MCD}$ computed from this minimizing subsample are called the raw minimum covariance determinant location and scatter estimates. The scatter part is scaled to achieve consistency for the covariance at the Gaussian distribution. 
Based on the raw estimates a reweighted scatter estimator $\hSigma_{\rm RMCD}$ is computed from the whole sample:
\[
	\hSigma_{\rm RMCD} = \left( \sum_{i =1}^n w_i \right)^{-1} \sum_{i =1}^n w_i (X_i - \hmu_{\rm MCD}) (X_i - \hmu_{\rm MCD})^T,
\]
where $w_i = 1$ if $(X_i - \hmu_{\rm MCD})^T \hSigma_{\rm MCD}^{-1} (X_i - \hmu_{\rm MCD}) < \chi^2_{p,1-\alpha}$ and zero otherwise, and $\alpha$ is a small rejection probability, e.g.\ $\alpha = 0\cdot05$. The reweighted covariance estimate is again scaled, but since this is not necessary for our applications we omit the details.

\section{Numerical example} \label{sec:NumEx}
\begin{figure}[b]
\centering
\includegraphics[scale=0.6]{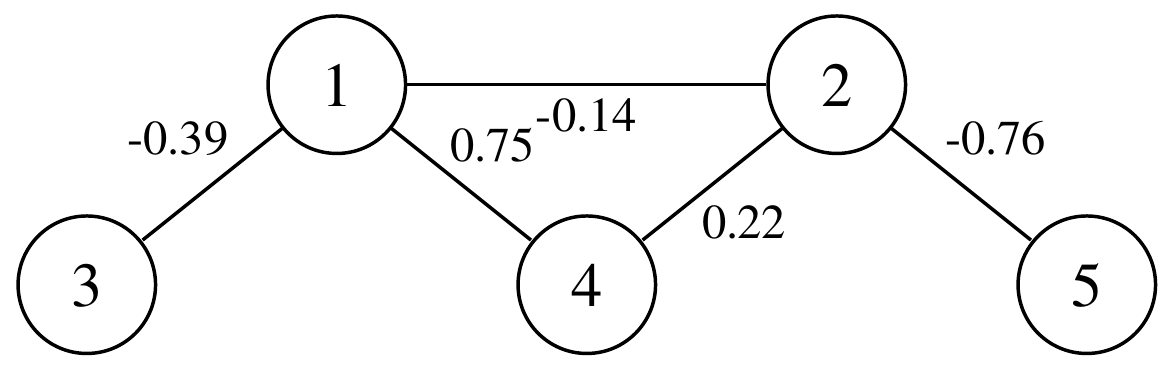}
\caption{Example model, edge labels indicate partial correlations} \label{Figure1}
\end{figure}
We present the results of a simulation study comparing several estimators.
We repeatedly sample 100 independent observations of a 5-dimensional distribution. We use the same shape matrix throughout, with equal diagonal elements and the partial correlation structure represented by the graph in Figure \ref{Figure1}. We let the tail behaviour vary, using the normal distribution and several members of the $t_{\nu,p}$ family to generate heavier tails \citep[][p.~207]{Bilodeau1999}. The index $\nu$ denotes the degrees of freedom. The moments of $t_{\nu,p}$ are finite up to order $\nu - 1$. We may talk of a fixed shape of the $t_{\nu,p}$ distribution, since $g$ is specified. For $\nu \ge 3$, its covariance matrix is $\nu(\nu-2)^{-1} S$, and, for $\nu \ge 5$, the excess kurtosis of each component is $6/(\nu-4)$. Propositions \ref{asymSigma} and \ref{asymTyler} imply that the Tyler matrix is asymptotically more efficient than the sample covariance matrix at $t_{\nu,p}$ if $\nu < p + 4$. 
For each distribution considered we generate 2000 samples, compute the estimates described in Section \ref{sec:Examples} and, based on each estimate, select a model.

We use a simple one-step model selection procedure, that allows us to concentrate on the effects of the different estimators. For each pair $\{i,j\}$ we test the model with all edges but $\{i,j\}$ against the saturated model, and exclude the edge $\{i,j\}$ if the test accepts the smaller model. The significance level $\alpha = 0\cdot05$ is an ad hoc choice. In our simulations the Wald-type test statistic $\hat{T}\!_n$ and the deviance test statistic $\hat{D}_n$ showed a very similar behaviour. Tables \ref{Table1} and \ref{Table2} report the results of the deviance test. 

\begin{table}[t] 
\small
\caption{One-step model selection based on $\hat{\Sigma}$ or $\hat{V}$  \label{Table1}}
\begin{tabular}{c@{\ \quad }c@{ \quad }|@{ \quad }r@{\ \quad }r@{\ \quad }r@{\ \quad }r} 
		distribution 
	& estimator	
	&  \parbox[b]{4.5em}{ \flushright mean edge difference}  
	&  \parbox[b]{6.5em}{ \flushright \% true \\ model found }
	&  \parbox[b]{6.5em}{ \flushright \%  non-edges \\ correctly found }
	&  \parbox[b]{6.5em}{ \flushright \%  \ding{172}${\not\relbar}$\ding{176} \\ correctly found } \\
\hline
	 normal		& $\hat{\Sigma}$         	&  	$1\cdot40$ 			& 21 & 79 & 95\\
 						& $\hat{\Sigma}^*$				&   $1\cdot41$ 			& 20 & 77 & 94\\	
 						& $\hat{V}$								&   $1\cdot65$    	& 14 & 78 & 94\\[0.5ex] 	
 	$t_{25}$	&	$\hat{\Sigma}$        	&  	$1\cdot44$ 			& 20 & 75	& 93\\
						& $\hat{\Sigma}^*$				&   $1\cdot44$			& 19 & 78	& 94\\
  					& $\hat{V}$   			 			&  	$1\cdot64$ 			& 14 & 78	& 94\\[0.5ex] 
	$t_{12}$	&	$\hat{\Sigma}$        	&  	$1\cdot51$ 			& 20 & 71	& 92\\
						& $\hat{\Sigma}^*$				&   $1\cdot51$ 			& 18 & 79 & 94\\
  					& $\hat{V}$   			 			&  	$1\cdot66$  		& 13 & 79 & 94\\[0.5ex] 
	$t_8$	    &	$\hat{\Sigma}$        	&  	$1\cdot65$ 			& 17 & 64	& 89\\
						& $\hat{\Sigma}^*$				&   $1\cdot65$			& 15 & 76	& 93\\
  					& $\hat{V}$   			 			&  	$1\cdot62$ 			& 13 & 79	& 94\\[0.5ex] 
	$t_5$	    &	$\hat{\Sigma}$        	&  	$1\cdot90$ 			& 14 & 51	& 84\\
						& $\hat{\Sigma}^*$				&   $1\cdot87$ 			& 10 & 74 & 93\\
  					& $\hat{V}$   			 			&  	$1\cdot63$  		& 14 & 78 & 94\\[0.5ex] 
  $t_3$			& $\hat{\Sigma}$ 					&   $2\cdot49$ 			&  8 & 29	& 72\\
						& $\hat{\Sigma}^*$				&   $2\cdot28$		  &  7 & 71 & 91\\
						& $\hat{V}$								&   $1\cdot65$ 			&  14& 78 &	95\\ 
\end{tabular}

\flushleft
$^*$ test statistic adjusted by estimated kurtosis
\end{table}

The main criterion by which we measure the goodness of the model selection is the mean edge difference, i.e., the average number of edges that are wrongly specified in the selected model, whether an existing edge was rejected or an absent edge was included. 
Although less suited as a performance criterion it is also of interest to know, how often the true model is found. 
Any model selection procedure that is based on testing for zero parameters aims at controlling the probability of correctly specifying the non-edges. We may also look at how often a single non-edge is correctly specified. This should be true in about 95\% of the cases, since a sample size of 100 seems large enough to expect some validity of the asymptotics in this setting.

In Table \ref{Table1} we compare the sample covariance matrix $\hSigma_n$ to Tyler's estimator $\hat{V}\!_n$ with the Hettmansperger--Randles median as location estimator. The benchmark is traditional graphical modelling, i.e., the performance of $\hSigma_n$ at the normal distribution. The classical deviance test deteriorates, if we move away from normality. We assume only ellipticity of the distribution and hence adjust the $\hSigma_n$-based test statistic by an estimate of $\sigma_1$, which is here the average of the sample kurtoses of all component divided by 3, cf.\ Proposition \ref{asymSigma}. This repairs the test, to some extent even in the case of the $t_3$-distribution, but does not necessarily give a better model selection. The estimator $\hSigma_n$ is inefficient under heavy tails, resulting in a test with low power.
As for Tyler's estimator, we recognize the asymptotic properties: the $\chi^2$-quantile fits, it outperforms $\hSigma_n$ at $t_{\nu}$-distributions with $\nu < 9$, and it is distribution-free within the elliptical model. 

In Table \ref{Table2} we examine if the same robustness against heavy-tailedness may be achieved by equally simple means using other robust estimators and, in particular, how the previous proposals of robust Gaussian graphical modelling, the reweighted minimum covariance determinant estimator and the Miyamura--Kano estimator, perform in this situation. 
Outlier-robust estimators interpret the bulk of the data as approximately normal and the observations in the tails as faulty outliers, that should be downweighted or rejected. Although there are some common aspects, this is in principle a different situation, and it is not surprising that both estimators do not meet the performance of Tyler's estimator at heavy-tailed distributions. Also, we did not estimate $\sigma_1$ from the data, but used its value for the normal distribution. For the reweighted minimum covariance determinant the values can be found in \citet{Croux1999}. But even in the Gaussian case, when $\sigma_1$ is chosen asymptotically correct, the asymptotic $\chi^2$-distribution does not seem to provide a sensible approximation. This small-sample inefficiency of the reweighted minimum covariance determinant estimator is usually taken care of by multiplying the test statistic by a correction factor, which has to be determined numerically \citep{Croux1999}. Using such an appropriate finite-sample value of $\sigma_1$ repairs the test, but again, it does not improve the model selection in our example. 
For the Miyamura--Kano proposal we note that they devise an alternative way of constrained estimation, but propose a very slow algorithm, which makes it, at least in the R implementation we used, unfeasible in larger dimensions. There is a tuning parameter to choose, which was set to $0\cdot3$ in our experiment, following the recommendation of the authors.
All calculations were done in R $2\cdot9\cdot1$, employing routines from the packages \verb"mvtnorm"\ , \verb"ggm"\ , \verb"ICSNP"\ , \verb"rrcov"\ and  \verb"rggm"\ . 
\begin{table}[t] 
\small
\caption{One-step model selection based on robust estimators\label{Table2}}
\begin{tabular}{c@{ \quad }c@{ \quad }|@{ \quad }r@{ \quad }r@{ \quad }r@{ \quad }r} 
		distribution 
	& estimator	
	&  \parbox[b]{4.5em}{ \flushright mean edge difference}  
	&  \parbox[b]{6.5em}{ \flushright \% true \\ model found }
	&  \parbox[b]{6.5em}{ \flushright \%  non-edges \\ correctly found }
	&  \parbox[b]{6.5em}{ \flushright \%  \ding{172}${\not\relbar}$\ding{176} \\ correctly found } \\
\hline
	 normal		& RMCD $0\cdot5$								&   $2\cdot05$ 				 	& 11 & 54 & 85\\	
 						& RMCD $0\cdot5^{**}$						&   $2\cdot06$    		 	&  5 & 81 & 94\\ 	
 						& RMCD $0\cdot75$								&   $1\cdot66$ 				 	& 15 & 72 & 92\\	
 						& RMCD $0\cdot75^{**}$					&   $1\cdot69$    		 	& 13 & 80 & 94\\ 	
 						& M--K$^+$											&   $1\cdot61$ 				 	& 14 & 81 & 95\\[0.5ex]	
	$t_3$		
						& RMCD $0\cdot5$								&   $2\cdot18$ 				 	&  9 & 45 & 82\\	
 						& RMCD $0\cdot5^{**}$						&   $2\cdot13$    		 	&  5 & 76 & 93\\ 	
 						& RMCD $0\cdot75$								&   $2\cdot02$ 				 	& 11 & 51 & 85\\	
 						& RMCD $0\cdot75^{**}$					&   $1\cdot96$    		 	& 10 & 61 & 89\\ 	
 						& M--K$^+$											&   $1\cdot82$ 				 	& 12 & 67 & 91\\	 
\end{tabular}

\flushleft
$^{**}$ with finite-sample correction; $^+$ \citet{Miyamura2006}
\end{table}

\section{Conclusion}
As a very simple and efficient technique to safeguard graphical modelling of continuous data against the impact of heavy tails, non-normality in general and, to some degree, also faulty outliers we recommend the use of Tyler's estimator in place of the empirical covariance matrix. The gain in robustness comes at a very moderate loss in efficiency, which becomes smaller with increasing dimension, and a justifiable increase in computing time. \citet{Vogel2010b} report average computing times on a 2.83 GHz Intel Core2 CPU for $n = 200$ and $p = 50$ of less than a second for the Tyler matrix, compared to less than three seconds for the reweighted minimum covariance determinant estimator.
Moreover, our approach allows the use of any affine pseudo-equivariant, root-$n$-consistent estimator $\hS\!_n$ in an analogous way. Assumption \ref{asymS} is the important prerequisite on $\hS\!_n$, and our results also apply to estimators that are asymptotically affine equivariant, like the rank-based estimation technique of \citet{Hallin2006}.

A problem that has not been addressed in this article is the accuracy of the asymptotic approximations for small to moderate sample sizes, in particular, to what extent it depends upon the ratio $p/n$.
This question splits into two parts. The first is an evaluation of the finite-sample properties of the affine pseudo-equivariant scatter estimators. These may be very different and do not allow a unified treatment. 
Very little seems to be known theoretically, either on the exact distribution of most robust scatter estimators or the rate of convergence to the Gaussian limit. 
However, there is strong empirical evidence that 
Tyler's estimator has excellent small-sample properties. In all our simulations the difference in the empirical distributions of any univariate function of the sample covariance matrix $\hat{\Sigma}_n$ at normality and the Tyler matrix $\hat{V}\!_n$ at any elliptical distribution, is fully expressed by the asymptotic scaling factor $1 + 2/p$, see also \citet[][Figure 2]{Vogel2010b}. Moreover, it is known that $\hat{V}\!_n$ behaves similarly to $\hat{\Sigma}_n$ when $p$ and $n$ grow large simultaneously \citep{Duembgen1998}. 
The second task is then, given the small-sample properties of the estimators, to assess the accuracy of the asymptotic $\chi^2$ distributions of the tests. This question is of relevance also in classical graphical modelling, where it has been noted that the deviance test statistic may substantially differ from its $\chi^2$ limit for small $n$. Improved small-sample approximations have been proposed \citep{Porteous1985, Porteous1989}, but also the exact distribution of the deviance test statistic is known for decomposable models, cf.\ \citet[][Sections 5.2.2 and 5.3.3]{Lauritzen1996}. Our simulations indicate that finite-sample correction techniques used in Gaussian graphical modelling may be put to good use also under ellipticity by applying it in an analogous way to Tyler's estimator.

The main limitation of the affine equivariant approach is that it does not provide a solution in the $p > n$ situation or allow a simple transfer of standard techniques, like regularization, that are used in Gaussian graphical modelling. Any affine equivariant, robust estimator requires more than $p + 1$ data points, because the only affine equivariant scatter estimator in the $p + 1 > n$ situation is the sample covariance estimator \citep{Tyler2010}. Dropping the affine equivariance property is inevitable for robust, high-dimensional graphical modelling.
 
\section*{Acknowledgement} This research was supported by the German Research Foundation. The authors gratefully acknowledge the assistance of Alexander D\"urre in preparing the figure and the simulations and thank the referees, the associate editor and the editor for their helpful comments and suggestions.

\appendix
\section{Proofs}

The proofs repeatedly apply the delta method to functions mapping matrices to matrices. We define the derivative of such a function, say, $g: \Rpp \to \Rpp$ at point $X$ as the derivative of $\vec\, g(X)$ with respect to $\vec\, (X)$ and denote its Jacobian at point $X$, which is of size $p^2 \times p^2$, by  $\dsD g(X)$. The symmetry of the argument poses a technical difficulty: there are $p(p+1)/2$ rather than $p^2$ variables, and the function $g$ must be viewed as a function from $\R^{p(p+1)/2}$ to $\Rpp$ in order to define a derivative. To deal with this issue we compute the Jacobian of $g$ interpreted as a function from $\Rpp$ to $\Rpp$ and post-multiply it by $M_p$. This is justified by the chain rule applied to $g = g_2 \circ g_1$, where $g_1$ duplicates the off-diagonal elements and $g_2: \Rpp \to \Rpp$. The derivatives below contain the right-multiplied $M_p$ depending on whether we view the function as defined on $\Ss_p$ or on $\Rpp$.
The textbook \citet{Magnus1999} covers most of the tools of the proofs, in particular calculation rules concerning the $\vec$ operator, the Kronecker product and derivatives of matrix functions. 
We repeatedly use the following without reference.
\begin{eqnarray*}
	(A \otimes B)(C \otimes D) = AC \otimes BD,  
	& & \
	(\vec A)^T \vec B = \trace(A^T B),
	\quad 
	\vec\,(ABC) = (C^T \otimes A) \vec B, \\
	M_p = M_p^2, & &  \quad M_p (A \otimes A) M_p = M_p (A \otimes A) = (A \otimes A) M_p, \quad  
\end{eqnarray*}
for matrices $A, B, C, D \in \Rpp$ \citep[][pp.~28, 30, 31]{Magnus1999}. Let $\iota: A \mapsto A\hme$ denote matrix inversion. Its Jacobian matrix is \citep[][p.~184]{Magnus1999} 
\[ 
	\dsD \iota(A) = -(A^T)^{-1} \otimes A^{-1}.
\]
\begin{proof}[Proof of Proposition \ref{Prop1}]
Part (\ref{asymK}) follows by straightforward calculations from the delta method. 

Part (\ref{asymP}): We have $\hPhS_n = \th(\hK_n)$ with $\th: A \mapsto - A_D\hmh A A_D\hmh$. We need to compute the derivative of $\th$ in order to apply the delta method. We start by considering $\th_0: A \mapsto A_D\hmh$. Its Jacobian matrix $\dsD \th_0(A) = -\frac{1}{2} \left\{A_D\hmh \otimes A_D\hme \right\}J_p$ 
is obtained by elementwise differentiation. Applying the multiplication rule to $\th(A) = - \th_0(A) A \th_0(A)$ yields
\be \label{th}
	\dsD \th(A) = - M_p \left\{\th(A) \otimes A_D\hme\right\}J_p \ - \ A_D\hmh \otimes A_D\hmh.
\ee
By the delta method,
\[
	n^{1/2}\vec\left(\hPhS_n - P \right) = n^{1/2}\vec\left\{ \th(\hK) - \th(\eta\hme K)\right\}
\]
converges in distribution to a $p^2$-dimensional normal distribution with mean zero and covariance matrix
\[
	\dsD\th(\eta\hme K) \eta^{-2} W_K(\sigma_1,\sigma_2) \left\{\dsD\th(\eta\hme K)\right\}^T,
\]
which reduces to the expression given in Proposition \ref{Prop1}. In particular, $\sigma_2$ vanishes, since $\dsD\th(K) \vec K = 0$. This is generally true for scale-invariant function $\th$.
\end{proof}
\begin{proof}[Proof of Proposition \ref{Prop2}]

Part (\ref{asymKG}): Since $K_G = \th_G(S)$ with 
\[
	\th_G: A \mapsto \sum_{k=1}^{2c-1} \zeta_k	\left(A_{\alpha_k,\alpha_k}^{-1} \right)^{(p)}
\]
we want to compute the derivative of $\th_G$. Let $\th_{\alpha}: A \mapsto (A_{\alpha,\alpha}\hme)^{(p)}$ for any subset $\alpha \subset \{1,\ldots,p\}$. The mapping $\th_{\alpha}$ is a composition of $(\cdot)_{\alpha,\alpha}$, $\iota$ and $(\cdot)^{(p)}$. We obtain by the chain rule
\[
	\dsD \th_{\alpha}(A) = - \left\{(A_{\alpha,\alpha}^{-1})^T \right\}^{(p)} \otimes \left(A_{\alpha,\alpha}^{-1} \right)^{(p)},
\qquad 
\quad
	\dsD \th_G(A) = - \sum_{k=1}^{2c-1} \zeta_k \left\{(A_{\alpha_k,\alpha_k}^{-1})^T \right\}^{(p)} \otimes \left(A_{\alpha_k,\alpha_k}^{-1} \right)^{(p)}.
\]
Then $\eta^{-2}W\!_{K_G} (\sigma_1,\sigma_2) = \dsD \th_G(\eta S) \eta^{2}W_S(\sigma_1,\sigma_2) \left\{ \dsD \th_G(\eta S) \right\}^T$ is shown to have the form given in Proposition \ref{Prop2} (\ref{asymKG}) by noting that
$\dsD \th_G(S) \vec S = \vec K_G$. This holds true because 
\[ 
 \left(S_{\alpha,\alpha}^{-1} \right)^{(p)} S \left(S_{\alpha,\alpha}^{-1} \right)^{(p)} = \left(S_{\alpha,\alpha}^{-1} \right)^{(p)},
\]
which is a consequence of the inversion formula for partitioned matrices.

Part (\ref{asymSG}): Applying the delta method we have to left- and right-multiply $W\!_{K_G}$ by the Jacobian of $\iota$ evaluated at $K_G$. Note that $(S_G \otimes S_G) \vec K_G = \vec S_G$.

Part (\ref{asymPG}): We left- and right-multiply $W\!_{K_G}$ by the Jacobian of $\th$, given in (\ref{th}), evaluated at $K_G$.
\end{proof}
\begin{proof}[Proof of Corollary \ref{Cor1}]
Let $S \in \Ss_p^+$ be such that $h_G(S) = S$ and write short $\Omega$ for $\Omega_G(S)$. 
It suffices to show that $2	M_p \Omega(S \otimes S)\Omega = 2 M_p\Omega$.
Proposition \ref{Prop2} (\ref{asymSG}) in connection with Proposition \ref{asymSigma} identifies the left-hand side as the asymptotic covariance of $h_G(\hat{\Sigma})$, where $\hat{\Sigma}$ is the sample covariance matrix, at the normal distribution with covariance $S$. Formula (5.50) in \citet{Lauritzen1996} identifies the same quantity as the right-hand side.
\end{proof}%
In the proofs of Proposition \ref{PropWald} and Corollary \ref{CorDev} we use the following lemma.
\bl \label{Lemma1}
Let $\X_n$ and $\X_n^{(m)}$, $m, n \in \N$, be as in Proposition \ref{PropWald} and $\hS\!_n$ a shape estimator such that $\hS\!_n(\X_n)$ satisfies Assumption \ref{asymS}. Assume furthermore that there is a continuously differentiable function $\xi\!: \R^{p \times p} \to \R$ with $\xi(I_p) = 1$ such that $\hS\!_n$ satisfies 
\be \label{ape.xi}
	\hS\!_n(\X_n A^T + 1_n b^T ) = \xi(AA^T) A \hS\!_n(\X_n) A^T
\ee 
for any data matrix $\X_n \in \R^{n \times p}$, $b \in \R^p$ and full rank matrix $A \in \R^{p \times p}$. Then
\[
	n^{1/2} \vec\,\left\{\hS\!_n(\X_n^{(n)}) - \eta S \right\} \, \to \, N_{p^2}\left\{\eta(B + c S), \, \eta^{2} W_S(\sigma_1,\sigma_2) \right\}
\]
in distribution as $n \to \infty$, where $B$ is as in Proposition \ref{PropWald} and $c = \dsD\xi(I_p)\vec\,(S\hmh B S\hmh)$.
\el
The proof of Lemma \ref{Lemma1} follows by straightforward calculations and is omitted. The constant
$c$ is identified by means of the first order Taylor expansion of $\xi(S_n\hh S\hme S_n\hh)$ around $I_p$.
\begin{proof}[Proof of Proposition \ref{PropWald}]
Part (\ref{PropWald1}) follows by standard arguments from the asymptotic normality of the estimator $\hP_{G_1}$. 
For any affine pseudo-equivariant estimator $\hS\!_n$ the rescaled estimator 
$\tilde{S}\!_n = \det(\hS\!_n)^{-1/p} \hS\!_n$ satisfies (\ref{ape.xi}), and the value of the test statistic $\hT_n(G_0,G_1)$ is the same, if computed from $\hS\!_n$ or $\tilde{S}\!_n$. Applying Lemma \ref{Lemma1} to $\tilde{S}\!_n$ we deduce part (\ref{PropWald2}) for analogously to part (\ref{PropWald1}).
\end{proof}
Towards the proof of Proposition \ref{PropDeviance} we state Lemmas \ref{LemmaMinimizer} to \ref{LemmaDerivatives}. For $A \in \Ss^+_p$ let $f_A: \Ss^+_p \to \R$: $f_A(B) =  \log \det B + \trace(B^{-1} A)$.	
From the theory of Gaussian graphical models we know that for any graph $G$ and $A \in \Ss^+_p$ the matrix $A_G = h_G(A)$ is the unique solution of the constrained optimization problem 
\be \label{cop1}
   \mbox{minimize } f_A(B) \qquad \mbox{ subject to } \quad  Q_{D(G)}\vec\,h(B) = 0, \ B \in \Ss^+_p,
\ee
because $A_G$ is the maximum likelihood estimate of the covariance matrix under the model $G$ at the normal distribution, if $A$ is the observed sample covariance, cf.\ \citet[][p.~133]{Lauritzen1996}.
Now with the notation of Section \ref{sec:Testing} let $H_0(\cdot) = Q_{D(G_0)}\vec\, h(\cdot)$, $H_1(\cdot) = Q_{D(G_1)}\vec\, h(\cdot)$ and $H_{0,1}(\cdot) = Q_{D(G_0)\setminus D(G_1)}\vec\, h(\cdot)$.
\begin{lemma} \label{LemmaMinimizer}
$A_{G_0} = h_{G_0}(A)$ is a solution of the constrained optimization problem
\be \label{cop2}
   \mbox{minimize } f_{A_{G_1}}\{h_{G_1}(C)\} \qquad \mbox{ subject to } \quad  H_{0,1}\{h_{G_1}(C)\} = 0, \ C \in \Ss^+_p.
\ee
\end{lemma}
\bop
By (\ref{cop1}) and (\ref{h_G}), $A_{G_0}$ uniquely solves the constrained optimization problem
\be \label{cop3}
 	\mbox{minimize } f_{A_{G_1}}(B)  \qquad \mbox{ subject to } \quad H_0(B) = 0, \ B \in \Ss^+_p. 
\ee
The restriction $H_0(B) = 0$ is equivalent to $H_1(B) = 0 \wedge H_{0,1}(B) = 0$, 
and any matrix $B$ with $H_1(B) = 0$ can be written as $B = h_{G_1}(C)$ for some  $C \in \Ss^+_p$. Thus $\left\{ B \ \middle| \ H_0(B) = 0, B \in \Ss^+_p \right\}$  and  
$	\Cc = \left\{ B = h_{G_1}(C) \ \middle| \ H_{0,1}\{h_{G_1}(C)\} = 0, C \in \Ss^+_p \right\}$ are equal, and so are the solution sets of the constrained optimization problems (\ref{cop3}) and
\be \label{cop4}
  	\mbox{minimize }	f_{A_{G_1}}(B)  \qquad \mbox{ subject to } \quad  B \in \Cc. 
\ee
Thus $A_{G_0}$ uniquely solves (\ref{cop4}), and all matrices $C \in \Ss^+_p$ with $h_{G_1}(C) = A_{G_0}$, among them $A_{G_0}$, solve (\ref{cop2}).
\eop
The next two lemmas are stated without proof. Expressions (\ref{Derivative2}) can be deduced from the proofs of Propositions \ref{Prop1} and \ref{Prop2}, and (\ref{Derivative1}) can be assembled from the derivatives given in \citet[][pp.~178,179]{Magnus1999}.
\begin{lemma} \label{LemmaAE}
Let $H: \Ss_p \to \R^q$ be continuously differentiable and $G_0$, $G_1$ as in Section \ref{sec:Testing}. Let furthermore $\hS\!_n$ be a sequence of almost surely positive definite random $p \times p$ matrices, for which $n^{1/2} (\hS\!_n - S)$ converges in distribution for some $S \in \Ss^+_p$ with $S^{-1} \in \Ss^+_p(G_0)$. Then for $n \to \infty$
\[
	n^{1/2} \left\{ H(\hS_{G_0}) -  H(\hS_{G_1}) \right\} \ \sim \ n^{1/2} \, \dsD H(\hS_{G_0}) \vec\,\left(\hS_{G_0} - \hS_{G_1}\right). 
\]
\end{lemma}
\begin{lemma} \label{LemmaDerivatives}
For $A, B \in \Ss^+_p$,
\begin{eqnarray}
\dsD f_A(B)  &  =  & \vec\,(B - A)^T (B^{-1} \otimes B^{-1}) M_p, \label{Derivative1} \\
\dsD h_{G}(B) &  =  & \left\{h_G(B) \otimes h_G(B) \right\} \Omega_G(B) M_p, \label{Derivative2} \quad
\dsD H_{0,1}(B)   =   Q_{0,1} \Gamma(B) \left(B^{-1} \otimes B^{-1}\right) M_p. 
\end{eqnarray}
\end{lemma}
\begin{proof}[Proof of Proposition \ref{PropDeviance}]
The second order Taylor expansion of $\log \det(\cdot)$ is 
\[
	\log \det (A + X) \ = \ \log \det A \ + \ \left\{\vec\, (A^T)^{-1}\right\}^T\vec X \ - \ \frac{1}{2} \left\{\vec\, (X^T) \right\}^T \left\{(A^T)^{-1} \otimes A^{-1}\right\} \vec X \ + \ o(||X||^2),
\]
cf.\ \citet[][pp.~108, 179, 184]{Magnus1999}. Applying this to the deviance test statistic yields
\begin{eqnarray}
\hD_n(G_0,G_1)  & \ = \ & n \left\{ \log \det(\hS_{G_0}) - \log \det(\hS_{G_1})\right\} 
									\ = \ - n \log \det \left( \hS_{G_1} \hS_{G_0}\hme \right) \nonumber \\
								& \ = \ & - n \trace\left(  \hS_{G_1} \hS_{G_0}\hme - I_p \right)
									\ + \ \frac{n}{2} \trace\left\{ \left( \hS_{G_1}\hS_{G_0}\hme - I_p \right)^2 \right\}
									\ + \ o\left( n || \hS_{G_1}\hS_{G_0}\hme - I_p ||^2 \right) \nonumber \\
							  & \  \sim \ & \frac{n}{2} \left\{\vec\,\left( \hS_{G_1} - \hS_{G_0} \right) \right\}^T
								  	\left( \hS_{G_0}\hme \otimes \hS_{G_0}\hme \right) 
								 				\vec\,\left( \hS_{G_1} - \hS_{G_0} \right),
								 				\qquad n \to \infty. \label{Schritt1}			
\end{eqnarray}
The asymptotic equivalence follows because
\begin{enumerate}[(1)]
\item $\trace\left(  \hS_{G_1} \hS_{G_0}\hme - I_p \right) 
				= \left\{ \vec\,\left(\hS_{G_1} - \hS_{G_0} \right) \right\}^T \vec \hS_{G_0}\hme \ = \ 0$, \
				which is a consequence of (\ref{h_G}), 
and  \\[0.1ex]			
\item $n || \hS_{G_1}\hS_{G_0}\hme - I_p ||^2 \, \le \, 
				\left(n^{1/2}||\hS_{G_1} -  S || + n^{1/2}||\hS_{G_0} -  S || \right)^2 ||\hS_{G_0}\hme ||^2
				= O_P(1), \qquad n \to \infty$. 
\end{enumerate}
Applying Lemma \ref{LemmaAE} to $H = h_{G_1}$ and using (\ref{Derivative2}) we find further
\[
	n^{1/2} \vec\,\left(\hS_{G_0} - \hS_{G_1} \right) \ \sim \ 
	n^{1/2} 	\left( \hS_{G_0} \otimes \hS_{G_0} \right) \Omega_{G_1}(\hS_{G_0}) M_p  \vec\,\left(\hS_{G_0} - \hS_{G_1} \right) 
\]
and from (\ref{Schritt1})
\be \label{Schritt2}
	\hD_n(G_0,G_1) \ \sim \ 
	\frac{n}{2} \left\{ \vec\,\left( \hS_{G_1} - \hS_{G_0} \right) \right\}^T
								   M_p \Omega_{G_1}(\hS_{G_0}) 
								  	\vec\,\left( \hS_{G_1} - \hS_{G_0} \right), \qquad n \to \infty.
\ee
Next we introduce the Lagrange multiplier \citep[][p.~131]{Magnus1999}. Since $\hS_{G_0}$ solves the constrained optimization problem (\ref{cop2}) with $A = \hS\!_n$, there exists a vector $\lambda \in \R^{q_{0,1}}$ such that
\[ 
						\dsD \left( f_{\hS_{G_1}} \circ h_{G_1} \right) \left(\hS_{G_0} \right) \ = \ 
	\lambda^T \dsD \left( H_{0,1}       \circ h_{G_1} \right) \left(\hS_{G_0} \right),
\] 
which transforms to $\displaystyle \
	M_p \Omega_{G_1}(\hS_{G_0}) \vec\,( \hS_{G_1} - \hS_{G_0}) \ = \ 
	M_p \Omega_{G_1}(\hS_{G_0}) \Gamma(\hS_{G_0})^T Q_{0,1}^T  \lambda$, \ 
cf.\ Lemma \ref{LemmaDerivatives}. \\[0.5ex]
We left-multiply both sides by $\,\hS_{G_0}\hh \otimes \hS_{G_0}\hh\,$ and solve for $\lambda$.
\begin{eqnarray*}
	& & 	M_p \Omega_{G_1}(\hS_{G_0}) \vec\,\left( \hS_{G_1} - \hS_{G_0} \right) \\ 
	& & = \ \ M_p \Omega_{G_1}(\hS_{G_0})  	\Gamma(\hS_{G_0})^T  Q_{0,1}^T 
	\left\{Q_{0,1} R_{G_1}(\hS_{G_0})  Q_{0,1}^T  \right\}^{-1}
	Q_{0,1} \Gamma(\hS_{G_0}) M_p \Omega_{G_1}(\hS_{G_0}) \vec\,\left( \hS_{G_1} - \hS_{G_0} \right).
\end{eqnarray*}
We substitute the right-hand side for the left-hand side of this equation in (\ref{Schritt2}), apply again Lemma \ref{LemmaAE}, this time to $H = H_{0,1} \circ h_{G_1}$, which leads to
\[
	n^{1/2} Q_{0,1} \vec \hP_{G_1} \ \sim \ 
	n^{1/2} Q_{0,1} \Gamma(\hS_{G_0}) M_p \Omega_{G_1}(\hS_{G_0}) \vec\,\left( \hS_{G_1} - \hS_{G_0} \right),
\]
and obtain
\[
	\hD_n(G_0,G_1) \ \sim \ 
	\frac{n}{2} \left(\vec \hP_{G_1}\right)^T Q_{0,1}^T 
	\left\{Q_{0,1} R_{G_1}(\hS_{G_0})  Q_{0,1}^T  \right\}^{-1}
	 Q_{0,1} \vec \hP_{G_1}, \qquad n \to \infty.
\]
Finally $R_{G_1}(\hS_{G_0}) \sim R_{G_1}(\hS\!_n)$ as $n \to \infty$, since both sides converge to $R_{G_1}(S)$. 
\end{proof}
\begin{proof}[Proof of Corollary \ref{CorDev}]
Part (1) is straightforward. For part (2) we take, as in Proposition \ref{PropWald}, the detour via $\tilde{S}\!_n = \det(\hS\!_n)^{-1/p} \hS\!_n$ and make use of Lemma \ref{Lemma1} to ensure that $\tilde{S}\!_n(\X_n^{(n)})$ meets the assumptions of Proposition \ref{PropDeviance}.
\end{proof}


\small


\begin{thebibliography}{30}
\providecommand{\natexlab}[1]{#1}
\providecommand{\url}[1]{\texttt{#1}}
\expandafter\ifx\csname urlstyle\endcsname\relax
  \providecommand{\doi}[1]{doi: #1}\else
  \providecommand{\doi}{doi: \begingroup \urlstyle{rm}\Url}\fi

\bibitem[Baba et~al.(2004)Baba, Shibata, and Sibuya]{Baba2004}
K.~Baba, R.~Shibata, and M.~Sibuya.
\newblock {Partial correlation and conditional correlation as measures of
  conditional independence.}
\newblock \emph{Aust. N. Z. J. Stat.}, 46\penalty0 (4):\penalty0 657--664,
  2004.

\bibitem[Becker(2005)]{Becker2005}
C.~Becker.
\newblock {Iterative proportional scaling based on a robust start estimator}.
\newblock In C.~Weihs and W.~Gaul, editors, \emph{{Classification - The
  Ubiquitous Challenge}}, pages 248--255. Heidelberg: Springer, 2005.

\bibitem[Bilodeau and Brenner(1999)]{Bilodeau1999}
M.~Bilodeau and D.~Brenner.
\newblock \emph{{Theory of Multivariate Statistics.}}
\newblock {New York, NY: Springer}, 1999.

\bibitem[Cox and Wermuth(1996)]{Cox1996}
D.~R. Cox and N.~Wermuth.
\newblock \emph{{Multivariate Dependencies: Models, Analysis and
  Interpretation.}}
\newblock {London: Chapman and Hall}, 1996.

\bibitem[Croux and Haesbroeck(1999)]{Croux1999}
C.~Croux and G.~Haesbroeck.
\newblock {Influence function and efficiency of the minimum covariance
  determinant scatter matrix estimator.}
\newblock \emph{J. Multivariate Anal.}, 71\penalty0 (2):\penalty0 161--190,
  1999.

\bibitem[D{\"um}bgen(1998)]{Duembgen1998}
L.~D{\"um}bgen.
\newblock {On Tyler's $M$-functional of scatter in high dimension.}
\newblock \emph{Ann. Inst. Stat. Math.}, 50\penalty0 (3):\penalty0 471--491,
  1998.

\bibitem[Edwards(2000)]{Edwards2000}
D.~Edwards.
\newblock \emph{{Introduction to Graphical Modelling.}}
\newblock {New York, NY: Springer}, 2000.

\bibitem[Fang and Zhang(1990)]{Fang1990b}
K.-T. Fang and Y.-T. Zhang.
\newblock \emph{{Generalized Multivariate Analysis.}}
\newblock {Berlin etc.: Springer-Verlag; Beijing: Science Press.}, 1990.

\bibitem[Friedman et~al.(2008)Friedman, Hastie, and Tibshirani]{Friedman2008}
J.~Friedman, T.~Hastie, and R.~Tibshirani.
\newblock {Sparse inverse covariance estimation with the graphical lasso.}
\newblock \emph{Biostatistics}, 9\penalty0 (3):\penalty0 432--441, 2008.

\bibitem[Gottard and Pacillo(2010)]{Gottard2010}
A.~Gottard and S.~Pacillo.
\newblock Robust concentration graph model selection.
\newblock \emph{{Comput. Statist. Data Anal.}}, 54\penalty0 (12):\penalty0
  3070--3079, 2010.

\bibitem[Hallin et~al.(2006)Hallin, Oja, and Paindaveine]{Hallin2006}
M.~Hallin, H.~Oja, and D.~Paindaveine.
\newblock {Semiparametrically efficient rank-based inference for shape. II:
  Optimal $R$-estimation of shape.}
\newblock \emph{Ann. Stat.}, 34\penalty0 (6):\penalty0 2757--2789, 2006.

\bibitem[Hettmansperger and Randles(2002)]{Hettmansperger2002}
T.~Hettmansperger and R.~Randles.
\newblock {A practical affine equivariant multivariate median.}
\newblock \emph{Biometrika}, 89:\penalty0 851--860, 2002.

\bibitem[Lauritzen(1996)]{Lauritzen1996}
S.~L. Lauritzen.
\newblock \emph{{Graphical Models.}}
\newblock {Oxford: Oxford Univ. Press}, 1996.

\bibitem[Magnus and Neudecker(1999)]{Magnus1999}
J.~R. Magnus and H.~Neudecker.
\newblock \emph{{Matrix Differential Calculus with Applications in Statistics
  and Econometrics.}}
\newblock {Chichester: Wiley}, 2nd edition, 1999.

\bibitem[Maronna et~al.(2006)Maronna, Martin, and Yohai]{Maronna2006}
R.~A. Maronna, D.~R. Martin, and V.~J. Yohai.
\newblock \emph{{Robust Statistics: Theory and Methods.}}
\newblock {Chichester: Wiley}, 2006.

\bibitem[Miyamura and Kano(2006)]{Miyamura2006}
M.~Miyamura and Y.~Kano.
\newblock {Robust Gaussian graphical modeling}.
\newblock \emph{J. Multivariate. Anal.}, 97\penalty0 (7):\penalty0 1525--1550,
  2006.
\newblock ISSN 0047-259X.

\bibitem[Ollila et~al.(2003)Ollila, Oja, and Croux]{Ollila2003}
E.~Ollila, H.~Oja, and C.~Croux.
\newblock {The affine equivariant sign covariance matrix: Asymptotic behavior
  and efficiencies.}
\newblock \emph{J. Multivariate Anal.}, 87\penalty0 (2):\penalty0 328--355,
  2003.

\bibitem[Ollila et~al.(2004)Ollila, Croux, and Oja]{Ollila2004}
E.~Ollila, C.~Croux, and H.~Oja.
\newblock {Influence function and asymptotic efficiency of the affine
  equivariant rank covariance matrix.}
\newblock \emph{Stat. Sin.}, 14\penalty0 (1):\penalty0 297--316, 2004.

\bibitem[Paindaveine(2008)]{Paindaveine2008}
D.~Paindaveine.
\newblock A canonical definition of shape.
\newblock \emph{Stat. Probab. Lett.}, 78\penalty0 (14):\penalty0 2240--2247,
  2008.

\bibitem[Porteous(1985)]{Porteous1985}
B.~Porteous.
\newblock {A note on improved likelihood ratio statistics for generalized log
  linear models.}
\newblock \emph{Biometrika}, 72:\penalty0 473--475, 1985.

\bibitem[Porteous(1989)]{Porteous1989}
B.~T. Porteous.
\newblock {Stochastic inequalities relating a class of log-likelihood ratio
  statistics to their asymptotic $\chi \sp 2$ distribution.}
\newblock \emph{Ann. Stat.}, 17\penalty0 (4):\penalty0 1723--1734, 1989.

\bibitem[Rousseeuw and Leroy(1987)]{Rousseeuw1987}
P.~J. Rousseeuw and A.~M. Leroy.
\newblock \emph{{Robust Regression and Outlier Detection.}}
\newblock {New York etc.: Wiley}, 1987.

\bibitem[Tyler(1982)]{Tyler1982}
D.~E. Tyler.
\newblock {Radial estimates and the test for sphericity.}
\newblock \emph{Biometrika}, 69:\penalty0 429--436, 1982.

\bibitem[Tyler(1983)]{Tyler1983}
D.~E. Tyler.
\newblock {Robustness and efficiency properties of scatter matrices.}
\newblock \emph{Biometrika}, 70:\penalty0 411--420, 1983.

\bibitem[Tyler(1987)]{Tyler1987}
D.~E. Tyler.
\newblock {A distribution-free M-estimator of multivariate scatter.}
\newblock \emph{Ann. Stat.}, 15:\penalty0 234--251, 1987.

\bibitem[Tyler(2010)]{Tyler2010}
D.~E. Tyler.
\newblock A note on multivariate location and scatter statistics for sparse
  data sets.
\newblock \emph{Stat. Probab. Lett.}, 80\penalty0 (17-18):\penalty0 1409--1413,
  2010.

\bibitem[Vogel and Fried(2010)]{Vogel2010a}
D.~Vogel and R.~Fried.
\newblock On robust {G}aussian graphical modelling.
\newblock In L.~Devroye, B.~Karas{\"o}zen, M.~Kohler, and R.~Korn, editors,
  \emph{Recent Developments in Applied Probability and Statistics. Dedicated to
  the Memory of J{\"u}rgen Lehn.}, pages 155--182. Berlin, Heidelberg:
  Springer-Verlag, 2010.

\bibitem[Vogel et~al.(2010)Vogel, D{\"u}rre, and Fried]{Vogel2010b}
D.~Vogel, A.~D{\"u}rre, and R.~Fried.
\newblock Elliptical graphical modeling in higher dimensions.
\newblock In \emph{{Proceedings of International Biosignal Processing
  Conference, July 14-16, 2010, Berlin, Germany.}}, pages 1--5, 2010.

\bibitem[Whittaker(1990)]{Whittaker1990}
J.~Whittaker.
\newblock \emph{{Graphical Models in Applied Multivariate Statistics.}}
\newblock {Chichester etc.: Wiley}, 1990.

\bibitem[Zuo(2006)]{Zuo2006}
Y.~Zuo.
\newblock Robust location and scatter estimators in multivariate analysis.
\newblock In J.~Fan and H.~Koul, editors, \emph{{Frontiers in Statistics.
  Dedicated to Peter John Bickel on Honor of his 65th Birthday}}, pages
  467--490. London: Imperial College Press, 2006.

\end{thebibliography}
\end{document}